\documentclass[notitlepage]{article}
\usepackage[titletoc,title]{appendix}
\usepackage{longtable}
\usepackage{fullpage}
\usepackage{times} 
\usepackage{amsmath}  
\usepackage{amssymb}
\usepackage{amsthm}
\usepackage{bbm}
\usepackage{bm}
\usepackage{float}
\restylefloat{table}
\restylefloat{figure}
\usepackage{rotating}
\usepackage{tablefootnote}
\usepackage[T1]{fontenc}
\usepackage{color}
\usepackage{breqn}
\usepackage{wrapfig}
\usepackage{amsfonts}
\usepackage{booktabs}
\usepackage{hyperref}
\usepackage[authoryear]{natbib}
\usepackage{graphicx}
\usepackage{mathtools}

\numberwithin{equation}{section}

\newtheorem{theorem}{Theorem}[section]

\DeclareMathOperator{\R}{\mathbb{R}}

\DeclareMathOperator{\mE}{\mathcal{E}}
\DeclareMathOperator{\mF}{\mathcal{F}}

\DeclareMathOperator{\mL}{\mathcal{L}}
\DeclareMathOperator{\mM}{\mathcal{M}}
\DeclareMathOperator{\mQ}{\mathcal{Q}}
\DeclareMathOperator{\mR}{\mathcal{R}}
\DeclareMathOperator{\mU}{\mathcal{U}}
\DeclareMathOperator{\mV}{\mathcal{V}}
\DeclareMathOperator{\mW}{\mathcal{W}}
\DeclareMathOperator{\Oo}{\mathcal{O}}
\DeclareMathOperator{\ad}{\text{ad}}

\newcommand{\half}{{\frac{1}{2}}}
\newcommand{\pd}{\partial}

\newcommand{\Sigmarr}{{\Sigma_{rr}}}
\newcommand{\Sigmazz}{{\Sigma_{zz}}}
\newcommand{\Sigmarz}{{\Sigma_{rz}}}

\newcommand{\BTT}{B^*(T_1,T_2)}
\newcommand{\BTTs}{{B^*}^2(T_1,T_2)}

\hypersetup{colorlinks=false, pdfborder={0 0 0}}

\begin{document}

\title{Analytic Pricing of SOFR Futures Contracts with Smile and Skew}
\author{Aurelio Romero-Berm\'{u}dez$^1$ \\ Colin Turfus$^2$}
\date{$^1$ING Bank N.V., Amsterdam, The Netherlands \\  $^2$ Independent\\[2ex] \today}
\maketitle

\begin{abstract}
	We introduce a perturbative formalism to solve the backward-looking futures pricing problem.
	The formalism is based on a time-ordered exponential series which allows to derive the functional form of the  integral kernel associated to the backward-Kolmogorov diffusion PDE. 
	We present  an analytic pricing formula for SOFR futures contracts under an extension of the Hull-White model
	which incorporates not only the intrinsic convexity adjustments captured by \cite{PricingSOFRFutures},
	but also the skew and smile observed in options markets as done in \cite{AnalyticRFROptions}.\footnote{The views expressed in this work represent personal research of the authors and do not necessarily reflect the views of their employers.}

\end{abstract}

\section{Overview}

SOFR is by now well on the way to replacing USD LIBOR as the benchmark rate used to determine floating rate payments in USD-denominated financial contracts, with USD LIBOR rates ceasing to be published from June 2023 \citep{EndOfLIBOR}.
This means the highly liquid Eurodollar futures market for 3M USD LIBOR will cease to function and be replaced by a SOFR futures market. Under the agreed fallback procedure, existing Eurodollar contracts will be switched from 3M USD LIBOR to 3M daily-compounded SOFR with a spread of 26.161bp, with options on this rate switched to options on SOFR with strike enhanced by 25bp \citep{SOFRFallback}. While the incorporation of the spread adjustment is straightforwardly managed, the fact that the interest rate fixing is backward- rather than forward-looking means that additional convexity considerations apply other than those associated with futures versus forward contracts.

The pricing of SOFR futures contracts faces issues in relation to the standard pricing theory for LIBOR underlyings. \cite{EurodollarConvexity} looked at the convexity adjustment which applies when pricing LIBOR-based Eurodollar futures contracts in comparison with the simpler (model-independent) forward contracts. They provide accurate analytic approximations for the industry standard LIBOR Market Model and show how these can be extended to incorporate the impact of smile and skew as captured from a stochastic volatility model, such as a Heston-style volatility.
Similarly, \cite{FutureIsConvex} have shown how the skew captured from the incorporation of a displaced diffusion can be accurately approximated .

\cite{PricingSOFRFutures} illustrates how the displaced diffusion approach can be extended in turn to incorporate modelling of OIS rates in a multi-curve framework. These OIS rates are taken to be governed by a Hull-White (HW) model \citep{H-W}. SOFR rates are then modelled as OIS rates plus a spread, usually taken to be deterministic. On this basis, \cite{PricingSOFRFutures} is able to evaluate 3M SOFR futures premia and so to infer an analytic expression for the associated convexity correction. However, although the displaced diffusion gives rise to skew effects in relation to associated LIBOR calculations, the use of a simple Hull-White model for SOFR futures calculations means that skew/smile effects are \emph{not} taken into account in SOFR futures pricing.
More recently, using Malliavin calculus, approximate semi-analytic expressions of the convexity of futures, FRAs and CMSs in various models of Cheyette type have been presented in \cite{Garcia_Merino_2023}.

The intention of the present work is to model SOFR rates instead through the recent Hull-White model extension of \cite{AnalyticRFROptions} which captures skew and smile effects through use of a local volatility which is a quadratic function of the short rate. In this way we are able to adjust SOFR futures prices not only for the inherent convexity, as was done by \cite{PricingSOFRFutures}, but also for the interest rate skew and smile observed in options markets. This is to our knowledge the first calculation attempting to capture these two effects together in an analytic expression.\footnote{During the completion of this work, \cite{Rakhmonov_quantminds} showed preliminary results on Eurodollar convexity in a stochastic volatility model.}
In order to derive such analytic expression, we employ a time-ordered exponential series to solve the backward Kolmogorov PDE associated to futures pricing.
Time-ordered exponential series have been employed over many years in areas such as Statistical Physics and Quantum Field Theory to find perturbative solutions.
In the context of Finance, the formalism is explained in detail in \cite{PerturbationMethods}.

In Sec. \ref{sec:ModelDescription}, we describe the dynamical model, which is a restatement of the model of \cite{AnalyticRFROptions},
and present the problem we solve together  with the solution strategy.
Then, in Sec. \ref{sec:GreensFunction}, we present the Green's function used as part of the solution strategy to solve the backward Kolmogorov PDE associated to futures pricing.
Calibration, which is based on  \cite{AnalyticRFROptions}, is briefly discussed in Sec. \ref{sec:calibration}.
The semi-analytic expressions for the 3M and 1M SOFR futures price and convexity incorporating smile and skew are given in Sec. \ref{sec:results}.
In Sec. \ref{sec:results_numerical}, we make explicit comparisons between the convexity extracted from Eurodollar and SOFR futures as well as
the impact of correctly capturing market skew and smile in convexity by comparing with the convexity extracted using a Hull-White model calibrated to the at-the-money
strike implied volatility of Caps.

\section{Model Description} \label{sec:ModelDescription}

We consider the short rate model of \cite{AnalyticRFROptions} which incorporates smile and skew into the Hull-White model through a variable transformation.
It is convenient to work with the reduced variable $y_t$  defined through the following Ornstein-Uhlenbeck process:
\begin{equation}
dy_t = -\alpha(t) y_t dt + \sigma(t)\,dW_t,\label{eq:dy_t}
\end{equation}
with $W_t$ a Wiener process for $t\ge0$. This auxiliary variable is related to the short rate by $r_t =r(y_t,t)$ where
\begin{align}
	r(y,t)&\coloneqq\overline{r}(t) + R^*(t) + \frac{\sinh \gamma(t)(y+y^*(t))}{\gamma(t)}.\label{eq:r_y_t}
\end{align}
Here $\overline{r},y^* :\mathbb{R}^+\to\mathbb{R}$ are the instantaneous forward rate and a skewness function, respectively, and $\sigma,\alpha, \gamma, R^*:\mathbb{R}^+\to\mathbb{R}^+$ are functions representing volatility, mean reversion rate, a smile factor and a convexity adjustment factor, respectively, all assumed to be piecewise continuous and bounded.
We further assume that $y_0=0$, with $t=0$ the ``as of'' date for which the model is calibrated. The function $R^*(t)$ is determined by calibration to the forward curve but will tend in the zero volatility limit to zero.%
\footnote{Strictly speaking, the function $y^*(\cdot)$ is the drift adjustment parameter which must be chosen to satisfy the no-arbitrage condition for a given choice of $R^*(t)$, so the latter should really be thought of as controlling the skew. However, the calculation of $R^*(t)$ as a power series of nested integrals involving $y^*(\cdot)$ is much more convenient than attempting to perform the equivalent process in reverse, so in practice we consider $y^*(t)$ to be inferred directly from the no-arbitrage condition, constructing the implied $R^*(t)$ using the mathematical techniques expounded below.}
The formal no-arbitrage constraint which determines this function is as follows
\begin{equation}
E\left[e^{-\int_{0}^tr_s\,ds}\right]=D(0,t)\label{eq:No_arbitrage}
\end{equation}
under the martingale measure for $0 < t\le T_m$, where $T_m$ is the longest maturity date for which the model is calibrated, and
\begin{equation}
D(t_1,t_2)\coloneqq e^{-\int_{t_1}^{t_2}\overline{r}(s)\,ds}
\end{equation}
is the $t_1$-forward price of the $t_2$-maturity zero coupon bond.
The explicit condition on $R^*(t)$ that satisfies Eq. \eqref{eq:No_arbitrage} has been given in \cite{AnalyticRFROptions}.
We re-use such calibration to the risk-neutral structure in the model presented here to price futures contracts.

Since we are concerned with future contracts with \emph{backward-looking} payoffs citing daily-compounded SOFR rates,
rather than contracts citing \emph{forward-looking} LIBOR rates or term-rates, we define a new variable $z_t$ representing the integral of the short rate:
\begin{equation}\label{eq:z_t}
	z_t\coloneqq \int_0^t (r_s - \overline{r}(s))ds.
\end{equation}

We consider the (stochastic) time-$t$ price $f(y_t,z_t,t)$ of a futures contract settled at maturity $T$ at a value $P(y_T,z_T,T)$, which is obtained as the risk-neutral expectation:
\begin{equation}\label{eq:PV_futures_def}
	f(y_t,z_t,t) = E\big[P(y_T,z_T,T)|\mF_t\big]
\end{equation}

We derive the functional form of $f(\cdot,\cdot,\cdot)$  from the Feynman-Kac theorem,
which states that  $f(y,z,t)$ emerges as the solution to the following backward Kolmogorov diffusion equation:
\begin{equation}\label{eq:futuresPDE}
\frac{\partial f}{\partial t}-\alpha(t)y\frac{\partial f}{\partial  y} +\left(r(y,t)- \overline{r}(t) \right)\frac{\partial f}{\partial z} +\frac12\sigma^2(t) \frac{\partial^2 f}{\partial y^2} = 0,
\end{equation}
with $r(y,t)$ given by \eqref{eq:r_y_t}, subject to the terminal condition $f(y,z,T)=P(y,z)$.
Note the absence of the discount term in the PDE Eq. \eqref{eq:futuresPDE} as a consequence of the absence of discounting in Eq. \eqref{eq:PV_futures_def}.
Depending on whether the futures references forward or backward rates, the payoff  may depend on state variables at $T$ only, or also at some $t\leq T$.
In particular, the futures payoff at $T_2$ citing daily-compounded SOFR rates for a payment period $[T_1,T_2]$
is approximated by continuous compounding as:\footnote{The compounding averaging is relevant for the 3M SOFR futures. For 1M SOFR futures, simple averaging is used. In both cases, the payoff depends on state variables at two times, see Sec. \ref{sec:results}.}
\begin{equation}\label{eq:P(y)}
P(z_{T_1},z_{T_2})=e^{\int_{T_1}^{T_2}r(x_t,t)dt}-1=D(T_1,T_2)^{-1}e^{z_{T_2}-z_{T_1}}-1.
\end{equation}
Therefore, the SOFR futures price $f(y_t,z_t,t)$ is calculated in two stages, in the first instance valuing the payoff as of time $T_1$, then treating the result as the payoff at time $T_1$ for a futures contract valued as of time $t<T_1$, as is standard practice with payoffs involving term rates.
On the other hand, for Eurodollar futures, for which we also present results, the price is calculated in a single stage since the payoff depends on state variables at $T_1$ only.

We seek, by analogy with the calculation of \cite{AnalyticRFROptions}, an exact Green's function solution for futures contracts with a payoff at some future time $T$ which depends upon the short rate and/or its integral. By this is meant a generalised function $G(x,z,t;\xi,\zeta,T)$ such that the futures price at time $t\ge0$ contingent upon $x_t=x$, $z_t=z$ of a contract which pays $P(x_T,z_T)$ at time $T$ is given by the convolution expression:
\begin{equation}\label{eq:convolution}
f(x,z,t)=\iint_{\R^2} G(x,z,t;\xi,\zeta,T)P(\xi,\zeta)d\xi d\zeta.
\end{equation}

\section{Green's Function}\label{sec:GreensFunction}
We repeat at this point the following definitions of \cite{AnalyticRFROptions}:
\begin{align}
	\phi_r(t,v)&\coloneqq e^{-\int_t^v \alpha(u)du},\label{eq:phi_r}\\
	\Sigma_{rr}(t,v)&\coloneqq\int_t^v\phi_r^2(u,v)\sigma^2(u)du,\label{eq:Sigma_rr}\\
	\psi_r(t,v)&\coloneqq e^{\frac12\gamma^2(v)\Sigma_{rr}(t,v)},\label{psi_r}\\
	\Sigma_{rz}(t,v)&\coloneqq\int_t^v \psi_r(t,u) \phi_r(u,v)\Sigma_{rr}(t,u)du, \label{eq:Sigma_rz}\\
	\Sigma_{zz}(t,v)&\coloneqq2\int_t^v \psi_r(t,u) \Sigma_{rz}(t,u)du, \label{eq:Sigma_zz}\\
	\bm{\Sigma}(t,v)&\coloneqq \left(\begin{array}{c c}
		\Sigma_{rr}(t,v) &\Sigma_{rz}(t,v)\\
		\Sigma_{rz}(t,v)&\Sigma_{zz}(t,v)
		\end{array}\right),\label{eq:Sigma}\\
	B^*(t,v)&\coloneqq\int_t^v \psi_r(t,u)\phi_r(t,u)du,\label{eq:B*}\\
	B^+(t,t_1,v)&\coloneqq\frac{B^*(t,v)-B^*(t,t_1)}{\phi_r(t,t_1)}\cr
	&=\int_{t_1}^v \psi_r(t,u)\phi_r(t_1,u)du,\label{eq:B+}\\
	\mu^*(y,t,v)&\coloneqq B^*(t,v)(y +\Sigma_{rz}(0,t)) +\textstyle\frac12B^{*2}(t,v)\Sigma_{rr}(0,t).\label{eq:mu*}\\
	V_C(t,u,v) & \coloneqq {B^*}^2(u,v)\Sigmarr(t,u)+\Sigmazz(u,v)
\end{align}
We observe here that $\phi_r(t,v)$ represents the impact of mean reversion and $\psi_r(t,v)$ that of the counteracting dispersion from the mean induced by smile.

In the absence of exact analytic solutions for Eq. \eqref{eq:futuresPDE},
we seek the Green's function as a perturbation expansion in the limit of small impact
of smile on the deviations of the short rate from the forward rate curve under some suitable norm.
To that end let us define the small parameter
\begin{equation}\label{eq:epsilon}
	\epsilon\coloneqq \left\|\gamma^2(\cdot)\Sigma_{rr}(\cdot,\cdot)\right\|
\end{equation}
and seek an asymptotic expansion of the Green's function in powers of $\epsilon$. To allow of a skew adjustment that contributes at the same order of approximation as the smile adjustment we shall make the assumption that $\|\gamma(\cdot)y^*(\cdot)\| = \Oo(\sqrt{\epsilon}).$%
\footnote{We have avoided being explicit about the range of $t$ over which the norms are assessed. A typical choice would be a weighted integral over $[0,T]$ for some finite $T$.}
Let us further express $R^*(\cdot)$ asymptotically as
\begin{equation}\label{eq:R*_expansion}
	R^*(t)=\sum_{n=1}^{\infty} R_n^*(t),
\end{equation}
with $R_n^*(\cdot)=\Oo(\epsilon^n)$.

As a consequence of the absence of the discount term in the pricing PDE, the following definitions deviate from those used in \cite{AnalyticRFROptions}
and are particular to the Green's function for futures:
\begin{equation}
	\begin{aligned}
		\mR_1(y,z,t,u)& \coloneqq R^+(y,t,u)\mM^+(y,z,t,u) - R^-(y,t,u)\mM^-(y,z,t,u)\,,\\
		R^{\pm}(y,t,u)&\coloneqq\frac{e^{\frac12\gamma^2(u)\Sigma_{rr}(t,u)\pm\gamma(u)(\phi_r(t,u)y +y^*(u))}} {2\gamma(u)}\,,
	\end{aligned}
\end{equation}

where $\mM^{\pm}$ are operators acting on real functions of $(y_t,z_t)$ by shifting their arguments:
\begin{align}
	\mM^{\pm}(y,z,t,u)g(y,z) &= g(y \pm \Delta y(t,u), z \pm \Delta z(t,u))\\
	\Delta y(t,u) & \coloneqq \gamma(u) \frac{\Sigmarr(t,u)}{\phi_t(t,u)}\,,\label{eq:Deltay}\\
	\Delta z(t,u) &\coloneqq \gamma(u)\Sigmarz(t,u) - B^*(t,u)\Delta y(t,u)\,.\label{eq:Deltaz}
\end{align}
\begin{theorem}[Green's Function]\label{PricingKernel}
	Making use of the above-defined notation, the Green's function for \eqref{eq:futuresPDE} can be written asymptotically in the limit as $\epsilon\to0$ as
\begin{equation}\label{eq:G_compounded}
	G(y,z,t;\eta,\zeta,v) = \sum_{n=0}^\infty G_n(y,z,t;\eta,\zeta,v)
\end{equation}
with $G_n(\cdot)=\Oo(\epsilon^n)$ and to first order
\begin{align}
	G_0(y,z,t;\eta,\zeta,v)&=N_2\Big(\eta-\phi_r(t,v)y; \zeta-z-\mu^*(y,t,v); \bm{\Sigma}(t,v)\Big), \label{eq:G_0}\\
	G_1(y,z,t;\eta,\zeta,v)&=\left[\int_t^v \Big(\mR_1(y,z,t,t_1) + R^*_1(t_1) \Big) dt_1 -\mQ(t,v) \right] \frac{\partial}{\partial z} G_0(y,z,t;\eta,\zeta,v), \label{eq:G_1}
\end{align}
with
\begin{align}\label{eq:Q}
	\mQ(t,v)&\coloneqq\mu^*(y,t,v)+\frac{\Sigma_{rz}(t,v)}{\phi_r(t,v)}\left(\frac{\partial}{\partial y}-B^*(t,v)\frac{\partial}{\partial z}\right) +\Sigma_{zz}(t,v) \frac{\partial}{\partial z}\,,
\end{align}
and $N_2(\cdot,\cdot;\bm{\Sigma})$ a bivariate Gaussian probability density function with covariance $\bm{\Sigma}$.
\end{theorem}
\begin{proof}
A proof of this result is given in Appendix \ref{app:PricingKernelProof} below.
\end{proof}

Theorem \ref{PricingKernel} and the formal expression given in Eq. \eqref{eq:kernel_formal} shows that the Green's function of the backward Kolmogorov equation \eqref{eq:futuresPDE} is expressed as a perturbation of the Hull-White Gaussian kernel modified by the skew and smile parameters.\footnote{The Hull-White kernel is $G_0$, the limit of \eqref{eq:G_compounded} as $\gamma\to0$.}
For practical purposes, it suffices to use the leading order deformation of the Hull-White kernel of \cite{AnalyticRFROptions}.

\section{Calibration}\label{sec:calibration}
The calibration to the observed term structure, Eq. \eqref{eq:No_arbitrage}, is performed using the zero-coupon bond formula,
which is obtained from the pricing PDE including the discount term as done in \cite{AnalyticRFROptions}.
This fixes $R^*_n(t)$ in \eqref{eq:R*_expansion} to the required level of accuracy.
For our purposes, we use the first order in $\epsilon$:
\begin{equation} \label{eq:R_1*}
	\begin{aligned}
		R_1^*(t)&=-R_1^+(0,0,t,t) +R_1^-(0,0,t,t)\cr
		&\quad-\psi_r(0,t)\int_0^t \phi_r(t_1,t)\Sigma_{rr}(0,t_1) (\gamma(t_1)(R_1^+(0,0,t_1,t)+R_1^-(0,0,t_1,t)) -\psi_r(0,t_1))dt_1\cr
		&=-\psi_r(0,t)\left(\frac{\sinh Y^*(t,t)} {\gamma(t)} +\int_0^t \psi_r(0,t_1) \phi_r(t_1,t) \Sigma_{rr}(0,t_1) \left(\cosh 	Y^*(t_1,t) -1\right)dt_1\right)\,,\\
		Y^*(t_1,t)&:=\gamma(t_1) \left(y^*(t_1) -B^+(0,t_1,t)\Sigma_{rr}(0,t_1) -\Sigma_{rz}(0,t_1)\right).
\end{aligned}
\end{equation}

The calibration of the rest of the parameters of the model, $\alpha(t)$, $\sigma(t)$, $\gamma(t)$ and $y^*(t)$
is attained using the caplet implied volatility formula of  \cite{AnalyticRFROptions}, which is  based also on the pricing kernel
of the PDE that includes the discount term.\footnote{Alternatively, calibration to backward-looking swaptions could also be done  using the implied volatility formula of \cite{AnalyticRFROptions}.}
For more details of the calibration, see \cite{AnalyticRFROptions}.

\section{SOFR Futures Contract Price}\label{sec:results}
Approximating the daily compounding of SOFR futures contracts by a continuous compounding,
the payoffs of 1M and 3M SOFR futures contracts may be written as:
\begin{equation}\label{eq:payoff}
	P(z_{T_1},z_{T_2})= \begin{cases}
		   \begin{aligned}
			   & \int\limits_{T_1}^{T_2}r(y_t,t)dt          &&= z_{T_2}{-}z_{T_1}-\log D(T_1,T_2)\,,    & T_2-T_1 = 1M\\
				 & e^{\int\limits_{T_1}^{T_2}r(x_t,t)dt}-1  &&= \frac{e^{z_{T_2}{-}z_{T_1}}}{D(T_1,T_2)} -1\,,     & T_2-T_1 = 3M
		   \end{aligned}
	\end{cases}
\end{equation}

\begin{theorem}[3M SOFR Future Price]\label{3MPrice}
	Consider a futures contract on the daily compounded risk-free rate over a borrowing period $[T_1,T_2]$, $T_2-T_1=3M$ in Eq. \eqref{eq:payoff}.
	The fair futures price for this contract is given asymptotically with relative errors $=\Oo(\epsilon^2)$ by
	\begin{equation}\label{eq:SOFR_3M}
	\begin{aligned}
		V_{\text{3M-SOFR}}(y,t)&\sim \left\{1 +  \int\limits_{T_1}^{T_2} dt_1 \left[R^+(y,t,t_1) {{\mM}}^+(y,z,T_1,t_1) - R^-(y,t,t_1) {\mM}^-(y,z,T_1,t_1) +R^*_1(t_1) \right] \right. \\
			&\left. \left. \phantom{\int\limits_{T_1}^{T_2}} -\mu^*\Big(\phi(t,T_1)y,T_1,T_2\Big)- V_C(t,T_1,T_2) \right\}  {e^{\mu^*\big(\phi(t,T_1)y,T_1,T_2\big)+ \half V_C(t,T_1,T_2) +z -z_{T_1}}\over D(T_1,T_2)}\right|_{z=z_{T_1}} -1\\
			&= {1+\Phi(y,t,T_1,T_2)\over D(T_1,T_2)}e^{\mu^*\big(\phi(t,T_1)y,T_1,T_2\big)+ \half V_C(t,T_1,T_2) } -1\,,
	\end{aligned}
\end{equation}
where
\begin{equation}\label{eq:Phi_def}
	\begin{aligned}
		\Phi(y,t,T_1,T_2) & := \int\limits_{T_1}^{T_2} dt_1\left\{ { \psi(t,t_1) \over \gamma(t_1)}\sinh\gamma(t_1)\Big[ \phi(t,t_1)y {+} y^*(t_1) {+} \BTT \phi(T_1,t_1) \Sigmarr(t,T_1) {+} \Sigmarz(T_1,t_1)\Big] +R^*_1(t_1)\right\}\\
		&-\mu^*\big(\phi(t,T_1)y,T_1,T_2\big)- V_C(t,T_1,T_2)\,.
	\end{aligned}
\end{equation}
\end{theorem}

\begin{proof}
	The proof of this result is set out in Appendix \ref{app:3M_PV}.
\end{proof}

\begin{theorem}[1M SOFR Future Price]\label{1MPrice}
	Consider futures contract on the daily compounded risk-free rate over a borrowing period $[T_1,T_2]$, $T_2-T_1=1M$ in Eq. \eqref{eq:payoff}.
	The fair futures price for this contract is given asymptotically with relative errors $=\Oo(\epsilon^2)$ by
\begin{equation}\label{eq:SOFR_1M}
	\begin{aligned}
		V_{\text{1M-SOFR}}(y,t)&\sim - \log D(T_1,T_2)  +\int\limits_{T_1}^{T_2} dt_1 \Big[R^+(y,t,t_1)  - R^-\big(y,t,t_1\big) +R^*_1(t_1) \Big] \\
		&= - \log D(T_1,T_2)  +\int\limits_{T_1}^{T_2} dt_1 \left\{ { \psi(t,t_1) \over \gamma(t_1)}\sinh\gamma(t_1)\Big[ \phi(t,t_1)y + y^*(t_1)\Big]  +R^*_1(t_1)\right\}
	\end{aligned}
\end{equation}
\end{theorem}

\begin{proof}
	The proof of this result is set out in Appendix \ref{app:1M_PV}.
\end{proof}

\subsection{Convexity}
Let us now consider the fair premium for the forward contract:
\begin{equation}
	V_{\text{Forward}}(y,t) = \frac{F^{T_1}(y,t)}{F^{T_2}(y,t)} -1 \simeq \frac{e^{-\mu^*(y,t,T_1)+\mu^*(y,t,T_2)}}{D(T_1,T_2)} \Big(1-\mF_1(y,t,T_1)+\mF_1(y,t,T_2)\Big)-1\,,\label{eq:Vfwd}
\end{equation}
where the zero-coupon bond price $F^T(y,t)$  and the functions $\mF_1(y,t,u)$ are derived in  \cite{AnalyticRFROptions}.
For convenience, we repeat these formulae here:
\begin{equation}\label{eq:F^T}
	\begin{aligned}
			F^T(y,t)& \sim D(t,T)e^{-\mu^*(y,t,T)}\left(1-\mF_1(y,t,T)\right)\,,\,
	\end{aligned}
\end{equation}
where
\begin{equation}\label{eq:Fcurly}
	\begin{aligned}
			\mF_1(y,t,T)& := \int\limits_t^T dt_1 \Big[\widetilde{R}_1^+(y,T_1,t_1,T_2)+\widetilde{R}_1^-(y,T_1,t_1,T_2)+R_1^*(t_1)\Big]-\mu^*(y,t,T) +\frac12\Sigma_{zz}(t,T)\,,\\
		&  = \int\limits_t^T dt_1\left\{ { \psi(t,t_1) \over \gamma(t_1)}\sinh\gamma(t_1)\Big[\phi_r(t,t_1)y {+} y^*(t_1) {-} B^+(t,t_1,T) \Sigmarr(t,t_1){-}\Sigmarz(t,t_1)\Big] +R^*_1(t_1)\right\}\\
		& -\mu^*(y,t,T) +\frac12\Sigma_{zz}(t,T)
	\end{aligned}
\end{equation}
and $\widetilde{R}_1^\pm$ are defined in Eq. \eqref{eq:operators_RFR_cap}.

For the 3M contract, the  difference between the SOFR fair values at $y=0$, $t=0$ and $V_{\text{Forward}}(0,0)$ gives the semi-analytic formula for
SOFR futures convexity in the extended Hull-White model that incorporates market smile and skew of Sec. \ref{sec:ModelDescription}:\footnote{Note that $\mF_1(0,0,T_1)=\mF_1(0,0,T_2)=0$ as required by the calibration condition $F^T(0,0) = D(0,T)$, see Theorem 4.1 of \cite{AnalyticRFROptions}.}
\begin{equation}\label{eq:3M_sofr_conv}
	\begin{aligned}
		C_{3M-SOFR} & \sim  {1+\Phi(0,0,T_1,T_2)\over D(T_1,T_2)}e^{\mu^*(0,T_1,T_2)+ \half V_C(0,T_1,T_2) }  - \frac{1}{D(T_1,T_2)} \\
	\end{aligned}
\end{equation}

The semi-analytic formula for the convexity of forward-looking futures in the same model is obtained analogously, see details in Appendix \ref{app:Eurodollar}.

For comparison purposes, we also calculate the convexity of backward- and forward-looking futures in the Hull-White model, which does not include market smile and skew, see Appendix \ref{app:HW}.

Since the 1M forward contract is not traded, we  \textit{define} the convexity of 1M SOFR futures by simply subtracting $-\log D(T_1,T_2)$ from the futures price, since this term represents the value of a contract paying the instantaneous forward rate for the borrowing period $[T_1,T_2]$.  In this way we ensure the convexity goes to zero in the limit of zero stochasticity. In other words,
\begin{equation}\label{eq:1M_sofr_conv}
	\begin{aligned}
		C_{1M-SOFR} &  \sim V_{\text{1M-SOFR}}(0,0)+ \log D(T_1,T_2)\\
		& = \int\limits_{T_1}^{T_2} dt_1 \psi(t,t_1)  {\sinh\gamma(t_1)y^*(t_1) - \sinh \gamma(t_1) (y^*(t_1) -\Sigmarz(0,t_1)) \over \gamma(t_1)}
	\end{aligned}
\end{equation}
We use the equivalent definition in the Hull-White model, where the 1M fair futures value is given in Eq. \eqref{eq:V_HW_comp}.

\section{Results}\label{sec:results_numerical}
In this section we show explicit results for the convexity of 3M and 1M SOFR futures, Eqs. \eqref{eq:3M_sofr_conv}, \eqref{eq:1M_sofr_conv}, in the model of Sec. \ref{sec:ModelDescription} that contains market smile and skew.
In addition, we compare these results with the convexity calculated in the Hull-White model, which neglects the effects of market smile and skew, see Appendix \ref{app:HW} for more details.
We calibrate the HW model at the ATM implied volatility of a caplet with the same duration, for which we use the analytical implied volatility formula of \cite{AnalyticRFROptions}.
We also compare the 3M SOFR futures convexity of  Eq. \eqref{eq:3M_sofr_conv} with the 3M Eurodollar convexity, where both include the effects of option smile and skew. For more details about the calculation of the Eurodollar convexity, see Appendix \ref{app:Eurodollar}.
These results are shown in Fig. \ref{fig:Convexity}.

\begin{figure}[H]
	\centering
	\includegraphics[width=0.42\linewidth]{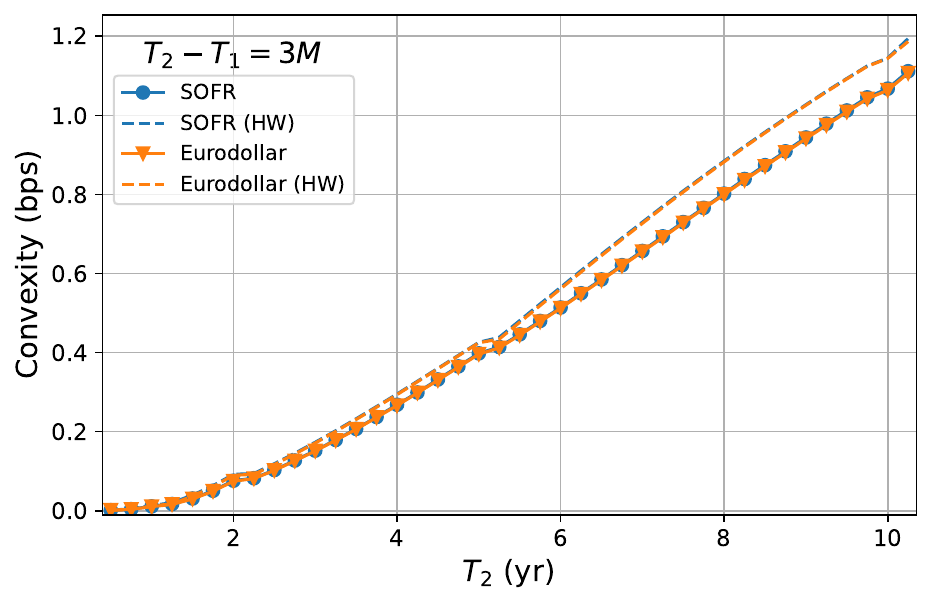}
	\includegraphics[width=0.42\linewidth]{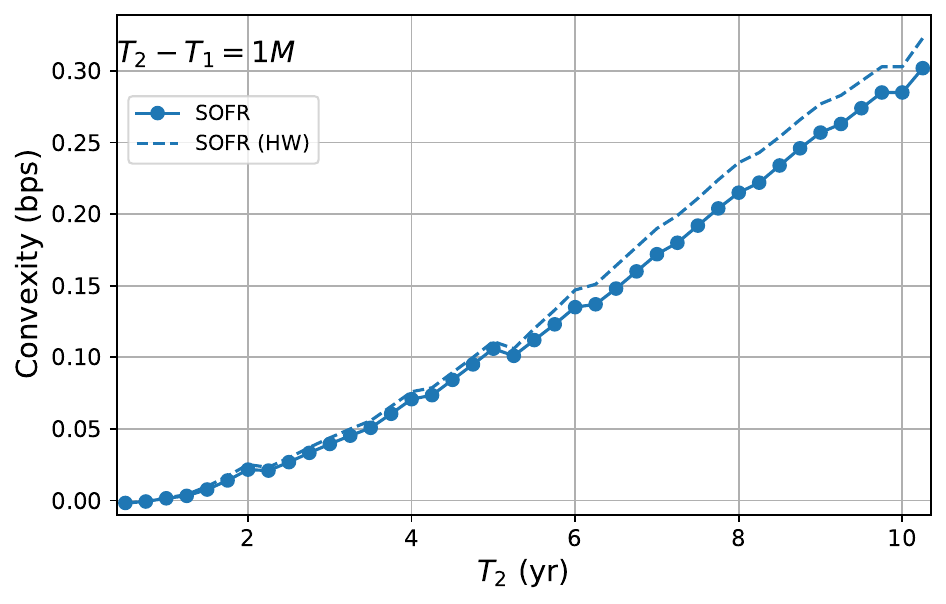}
	\vspace{-2mm}
	\caption{Convexity at $y_t=0$, $t=0$  for 3M SOFR and Eurodollar convexity (left, solid lines) and   1M SOFR convexity (right, solid line).
	These solid lines with dots are calculated from Eqs. \eqref{eq:3M_sofr_conv} and \eqref{eq:1M_sofr_conv} for SOFR contracts and the solid line with triangles using Eq. \eqref{eq:V_Eurodollar} for the fair Eurodollar value.
	The dashed lines correspond to the convexity calculated in the Hull-White model using Eqs. \eqref{eq:V_HW_comp} and \eqref{eq:V_HW}. This model is calibrated at the ATM implied volatility of a caplet with the same contract duration.}\label{fig:Convexity}
\end{figure}

Fig. \ref{fig:impact_smile} shows the impact of not modelling correctly the option smile and skew on the futures convexity may be around $20\%$ for short maturities.
In this figure, the convexity in the model of Sec. \ref{sec:ModelDescription}, which is obtained from Eqs. \eqref{eq:3M_sofr_conv}  and \eqref{eq:1M_sofr_conv}
is compared with the convexity in the Hull-White model, see Appendix \ref{app:HW} for more details.
The Hull-White model used for this purpose has been calibrated at the at-the-money implied volatility of caplets with the same duration. This captures the level but not the shape (smile and skew)
of the implied volatility surface.
\begin{figure}[H]
	\centering
	\includegraphics[width=0.42\linewidth]{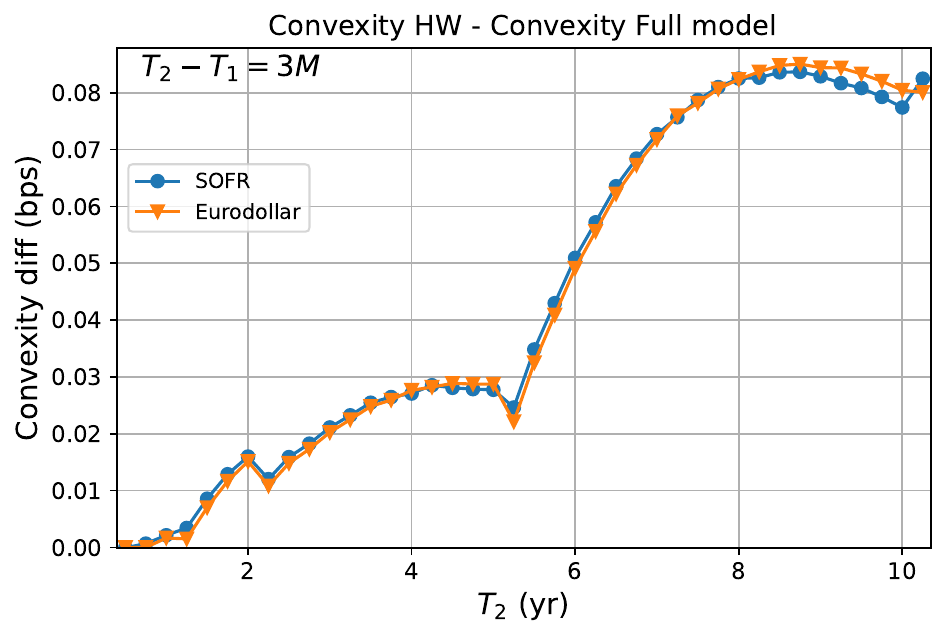}
	\includegraphics[width=0.42\linewidth]{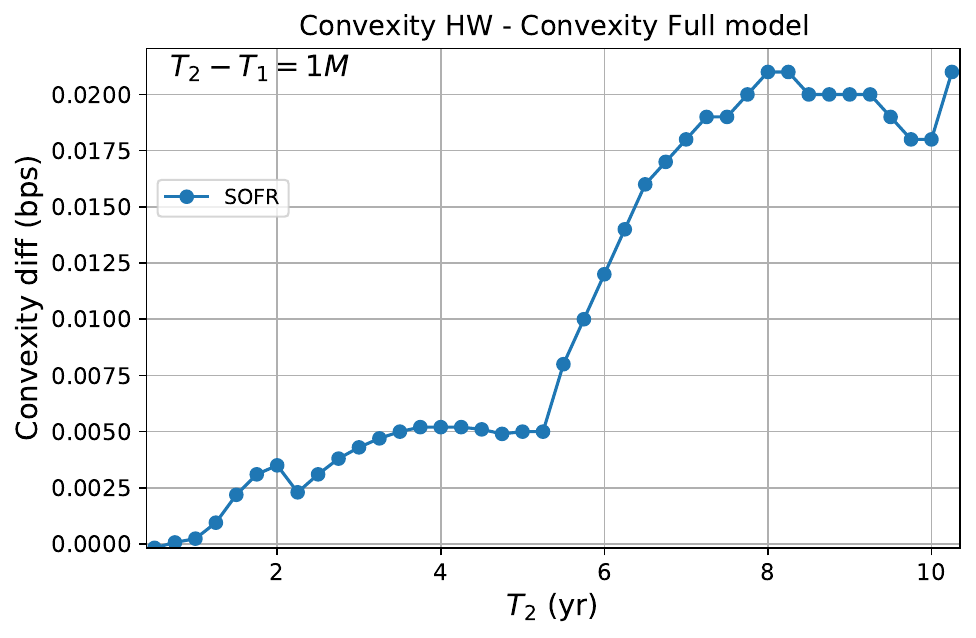}

	\includegraphics[width=0.42\linewidth]{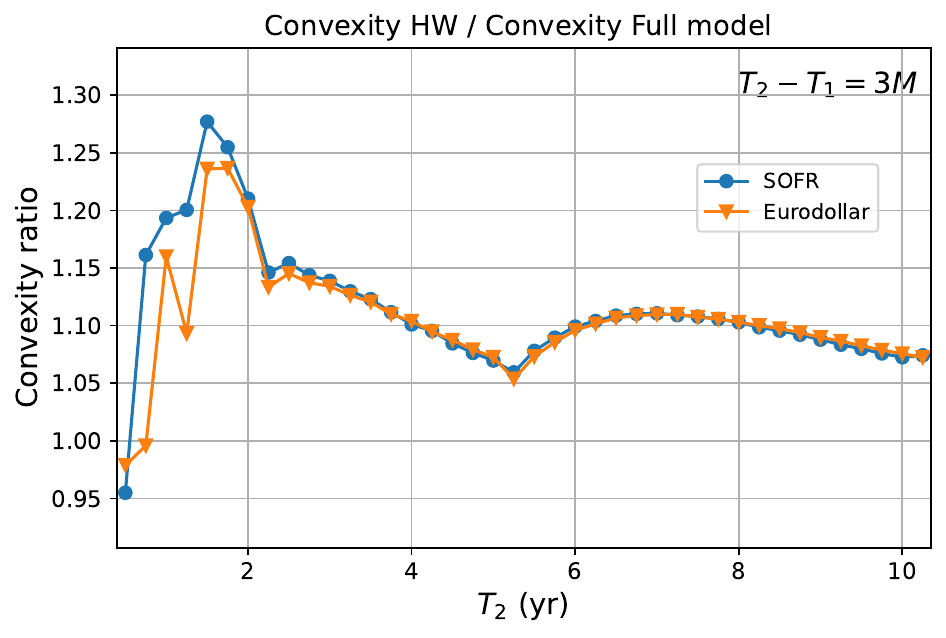}
	\includegraphics[width=0.42\linewidth]{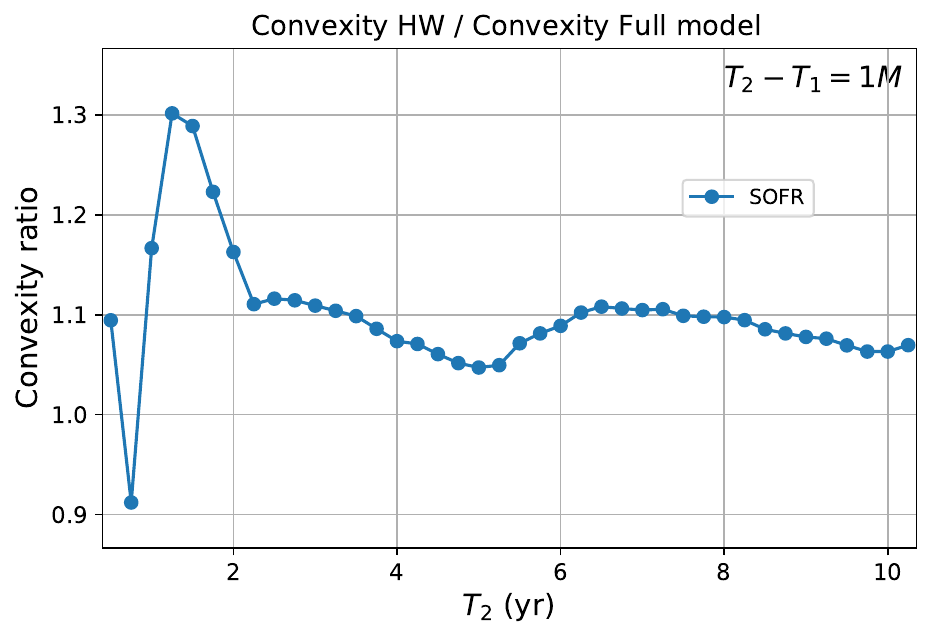}

	\vspace{-2mm}
	\caption{Absolute and relative differences of the convexity at $y_t=0$, $t=0$ of SOFR and Eurodollar contracts including the effect of option smile and skew and the convexity obtained in the Hull-White model,
		which captures smile and skew only at the ATM strike. The difference is calculated as the Hull-White convexity minus the full model convexity.
		The ratio is calculated dividing the difference by the full model convexity.}\label{fig:impact_smile}
\end{figure}
Finally, we show the difference between the convexity of 3M SOFR futures and the convexity of the Eurodollar future in Fig. \ref{fig:cmp_fwd}.
This comparison shows that the convexity in 3M forward- and backward-looking futures is very similar in absolute terms,
although in relative terms the convexity of forward-looking futures is between $10-15\%$ smaller for short settlement dates.
\begin{figure}[H]
	\centering
	\includegraphics[width=0.42\linewidth]{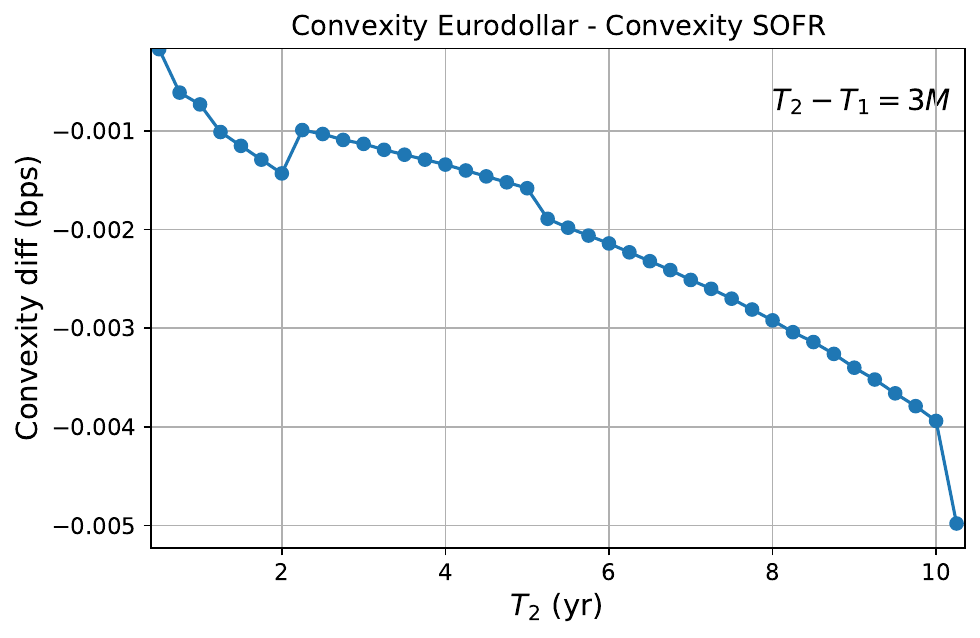}
	\includegraphics[width=0.42\linewidth]{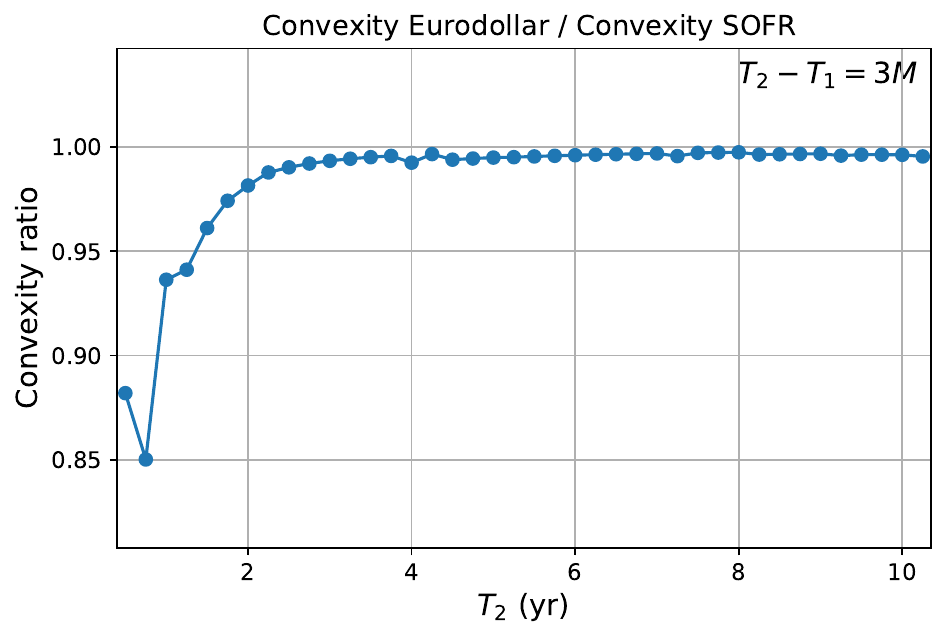}

	\vspace{-2mm}
	\caption{Convexity of 3M-Eurodollar minus the convexity of 3M-SOFR contracts, both incorporating option smile and skew.
	    The Eurodollar convexity is calculated using  Eq. \eqref{eq:V_Eurodollar}, as explained in Appendix \ref{app:Eurodollar}.
		The ratio is calculated dividing by the 3M-SOFR convexity of Eq. \eqref{eq:3M_sofr_conv}.}\label{fig:cmp_fwd}
\end{figure}

\section{Conclusions and Future work}
In this work, we have introduced the time-ordered exponential series formalism to obtain analytical pricing formulae of forward and backward-looking futures in a short rate model incorporating volatility smile and skew.
We have shown that if market smile and skew is incorporated only by calibrating the Hull-White model to at-the-money implied volatilities,
convexity is overestimated by $10-25\%$, depending on settlement date.
We also show the difference in convexity between forward- and backward-looking futures is very small in the 3M future in absolute terms, although in relative terms the Eurodollar convexity is between
$10-15\%$ smaller for short settlement dates.

We observe in passing that, for SOFR mid-curve options, i.e.~options to enter into a futures contract with strike $K$ and an exercise date more than a year in the future, the effective payoff for exercise at time $\tau$ is
\begin{equation}\label{eq:V_option}
	P_{\text{option}}(x_\tau,\tau) = V_{\text{Futures}}(x_\tau,\tau) -K\delta(T_1,T_2)
	\,.
\end{equation}
This expression is potentially useful for the pricing of such options.

\addcontentsline{toc}{section}{References}

\bibliography{Bibliography}

\appendix

\section{Proof of the Green's function expression}\label{app:PricingKernelProof}
Here, we outline the key steps used in the derivation of the Green's function of Theorem \ref{PricingKernel}. This derivation is very similar to the one presented in \cite{AnalyticRFROptions}, except that
the backward Kolmogorov pricing PDE, Eq. \ref{eq:futuresPDE}, now does not include a discount term.
The starting point is the calculation of the semi-group or evolution operator associated to the pricing PDE Eq. \eqref{eq:futuresPDE}:
\begin{equation}
\begin{aligned}
    \mU(t,v) &= \mE_t^v(\mL_0(\cdot)  + \mV_0(\cdot))\,,\\
	\mL_0(t)&=-\alpha(t)y\frac{\partial}{\partial  y} +\frac12\sigma^2(t) \frac{\partial^2 }{\partial y^2} ,\\
	\mV_0(t)&= \left(r(y,t)- \overline{r}(t) \right)\frac{\partial}{\partial z}\,,
\end{aligned}
\end{equation}
where $ \mE_t^v(\mL)$ represents the operator given by the time-ordered exponential series of the operator $\mL$, see \cite{AnalyticRFROptions} for more details.
The associated Green's function is obtained following the same steps used in the derivation of the Green's function in \cite{AnalyticRFROptions}:
\begin{equation}\label{eq:G_intermediate}
	\begin{aligned}
		G(y,z,t; \eta,\zeta,v) &= \mE_t^v(\mW_1(t,\cdot) ) N_2(\eta-\phi_r(t,v); \zeta-z) \,,\\
		\mW_1(t,u)&\coloneqq\mE_t^u\left(\ad_{\mL_1(t,\cdot)}\right)(\mV_1(t,u))\\
		\mL_1(t,u) & \coloneqq \psi_r(t,u)\left[ \frac{\Sigmarr(t,u)}{\phi_r(t,u)}\pd_y + \Sigmarz(t,u)\pd_z\right]\\
		\mV_1(t,u) &\coloneqq \left[R^+(y,t,u)e^{\Delta y(t,u)\pd_y} - R^-(y,t,u)e^{-\Delta y(t,u)\pd_y} + R^*_1(u)\right]\pd_z - \mL_1(t,u)\\
		\Delta y(t,u) & \coloneqq \gamma(u) \frac{\Sigmarr(t,u)}{\phi_t(t,u)}\,,\\
		R^{\pm}(y,t,u)&\coloneqq\frac{e^{\frac12\gamma^2(u)\Sigma_{rr}(t,u)\pm\gamma(u)(\phi_r(t,u)y +y^*(u))}} {2\gamma(u)}\,.
	\end{aligned}
\end{equation}

The operator $\mE_t^v(\mW_1(t,\cdot) )$ in Eq. \eqref{eq:G_intermediate} is conveniently simplified as follows:
\begin{equation} \begin{aligned}
    \mE_t^v(\mW_1(t,\cdot) ) &= \mE_t^v(\mW_2(t,\cdot) ) \mE_t^v(\mL_2(t,\cdot) )\\
	\mW_2(t,u)&\coloneqq\mE_t^u\left(\ad_{\mL_2(t,\cdot)}\right)(\mW_1(t,u)-\mL_2(t,u))\\
	\mL_2(t,u) & \coloneqq \pd_u \mu^*(y,t,u) \pd_z\,,
\end{aligned} \end{equation}
which leads to:
\begin{equation} \begin{aligned}
    \mE_t^v(\mW_1(t,\cdot) ) &= \mE_t^v(\mW_2(t,\cdot) ) e^{\mu^*(y,t,v) \pd_z}\,,\\
	\mW_2(t,u) &= \left[ 	\mR_1(y,t,u) + R^*(u)\right]\pd_z - \mL_3(t,u)\\
	\mR_1(y,z,t,u)& \coloneqq R^+(y,t,u)e^{\Delta y(t,u)\pd_y+ \Delta z(t,u)\pd_z} - R^-(y,t,u)e^{-\Delta y(t,u)\pd_y - \Delta z(t,u)\pd_z}\\
	\mL_3(t,u) &\coloneqq \mL_1(t,u)+\mL_2(t,u) - \psi_r(t,u)B^*(t,u)\frac{\Sigmarr(t,u)}{\phi_r(t,u)}\pd^2_{z}\\
	\Delta z(t,u) &\coloneqq \gamma(u)\Sigmarz(t,u) - B^*(t,u)\Delta y(t,u)\,.
\end{aligned} \end{equation}
The formal expansion of the Green's function is:
\begin{equation}\label{eq:kernel_formal}
	G(y,z,t;\eta,\zeta,v)=\mE_t^v(\mW_2(t,\cdot))N_2\left(\eta-\phi_r(t,v)y; \zeta-z-\mu^*(y,t,v)\right)\,,
\end{equation}
where the leading term corresponds to the Green's function of the Hull-White model, including the compounding variable $z$, defined in Eq. \eqref{eq:z_t}.
The subleading term containing the market skew and smile information is:
\begin{equation} \begin{aligned}
	G_1(y,z,t;\eta,\zeta,v)& = \left[\int\limits_t^v dt_1 \Big[ 	\mR_1(y,t,t_1) + R^*_1(t_1)\Big]  - \mQ(y,t,v)\right]\pd_z G_0(y,z,t;\eta,\zeta,v)\,,\label{eq:kernel_O1}\\
	\mQ(y,t,v)&\coloneqq \mu^*(y,t,v)+ \frac{\Sigmarz(t,v)}{ \phi_r(t,v) }\Big(\pd_y - B^*(t,v)\pd_z\Big) + \Sigmazz(t,v)\pd_z\quad \,.
\end{aligned} \end{equation}

The function $R^*_1(t)$ is obtained from the of today's observed forward curve, see Sec. \ref{sec:calibration}.

\section{3M SOFR futures price}\label{app:3M_PV}
Here, we derive the fair value of the futures payoff of Eq. \eqref{eq:payoff} with continuous compounding ($T_2-T_1=3M$ case) in the model presented in Sec. \ref{sec:ModelDescription}.
As mentioned in Sec. \ref{sec:ModelDescription}, the fair value is calculated in two stages:
a first integral involving the Green's function between $T_1$ and $T_2$, and a second integral involving the Green's function between $t$ and $T_1$.
The zeroth-order contribution to the fair value is:
\begin{equation}\label{eq:SOFR_V0}
	\begin{aligned}
		V^{(0)}(y,T_1) &= \int d\eta d\zeta P(z_{T_1},z_{T_2}) G_0(y,z,T_1;\eta,\zeta,T_2 )= {e^{\mu^*(y,T_1,T_2)+\half \Sigmazz(T_1,T_2)}\over D(T_1,T_2)} -1 \,,\\
		V^{(0)}(y,t) &= \int d\eta d\zeta V^{(0)}(\eta,T_1) G_0(y,z,t;\eta,\zeta,T_1 ) = {e^{\mu^*(\phi(t,T_1)y,T_1,T_2)+ \half V_C(t,T_1,T_2)}\over D(T_1,T_2)} -1\,.
	\end{aligned}
\end{equation}
The first order price at $T_1$ is:
\begin{equation} \begin{aligned}
	V^{(1)}(y,T_1) &= \lim_{z\to z_{T_1}} \int d\eta d\zeta \left({e^{\zeta-z_{T_1}}\over D(T_1,T_2)}-1\right) G_1(y,z,T_1;\eta,\zeta,T_2 )\nonumber \\
	& = \lim_{z\to z_{T_1}} \left[ \int\limits_{T_1}^{T_2} dt_1  \Big[R^+(y,T_1,t_1) \mM^+(T_1,t_1) - R^-(y,T_1,t_1)\mM^-(T_1,t_1)  + R^*_1(t_1)\Big]\right.\nonumber \\
	&\hspace{1cm} \left. \phantom{\int\limits_{T_1}^{T_2}}- \mu^*(y,T_1,T_2) - \Sigmazz(T_1,T_2)  \right] {e^{\mu^*(y,T_1,T_2)+z-z_{T_1}+\half \Sigmazz(T_1,T_2)}\over D(T_1,T_2)}\,,\label{eq:V1T1}
\end{aligned} \end{equation}
where $y$ represents the OU state variable at $T_1$ and $R^*_1$ is given in Eq. \eqref{eq:R_1*}.

The second stage involving the integral of  $V(y,T_1)$ and the propagator between $t$ and $T_1$ has, in principle, two contributions.
One contribution originates from the zeroth-order term of $V(y,T_1)$ and the first-order term of the Green's function. This contribution vanishes:
\begin{equation*}
	\begin{aligned}
		V^{(1,0)}(y,t) &= \int d\eta d\zeta V^{(0)}(\eta,T_1) G^{(1)}(y,z,t; \eta, \zeta, T_1) \\
		&= \left[ \int\limits_{T_1}^{T_2} dt_1 \left[\mR_1(y,z,t,t_1)  +R^*_1(t_1) \right] - \mQ(t,T_1)\right]\pd_z\int d\eta d\zeta V^{(0)}(\eta,T_1) G^{(0)}(y,z,t; \eta, \zeta, T_1) \\
		& =0\,,
	\end{aligned}
\end{equation*}
because the integrals in $\eta,\zeta$ give $V^{(0)}(y,T_1)$, which is independent of $z$.
The second contribution originates from the zeroth-order term of the Green's function  and the first-order term of $V(y,T_1)$ in Eq. \eqref{eq:V1T1}:
\begin{equation}
	\begin{aligned}
		V^{(0,1)}(y,t) &= \int d\eta d\zeta G^{(0)}(y,z,t; \eta, \zeta, T_1) V^{(1)}(\eta,T_1)\,,
	\end{aligned}
\end{equation}
which results in Eq. \eqref{eq:SOFR_3M}.

\section{1M SOFR futures price }\label{app:1M_PV}
A similar calculation to that of the 3M SOFR futures in Sec. \ref{app:3M_PV} gives, at leading order, the fair value of the 1M futures contract, defined with simple average of the interest rates:
\begin{equation} \label{eq:SOFR_V0_1M}
	\begin{aligned}
		V^{(0)}(y,T_1)  & =- \log D(T_1,T_2)+ \lim_{z\to z_{T_1}} \int d\eta d\zeta (\zeta-z_{T_1})G_0(y,z,T_1;\eta,\zeta,T_2 )  \\
						& = - \log D(T_1,T_2) + \mu^*(y,T_1,T_2)\,.\\
		V^{(0)}(y,t) &= \int d\eta d\zeta V^{(0)}(\eta,T_1) G_0(y,z,t;\eta,\zeta,T_1 )\\
					  &= - \log D(T_1,T_2) + \mu^*(y\phi(t,T_1),T_1,T_2)\,.
	\end{aligned}
\end{equation}
At time $T_1$, the price's subleading term is:
\begin{equation*}
	\begin{aligned}
		V^{(1)}(y,T_1)  & = \lim_{z\to z_{T_1}}  \left[\int\limits_{T_1}^{T_2} dt_1 \left[ 	\mR_1(y,z,t,t_1) + R^*_1(t_1)\right]  - \mQ(y,T_1,T_2)\right]\pd_z \int d\eta d\zeta P_{1M-\text{fut}}(z_{T_1}, \eta) G_0(y,z,T_1;\eta,\zeta,T_2 )  \\
						& =  \lim_{z\to z_{T_1}}  \left[\int\limits_{T_1}^{T_2} dt_1 \left[ 	\mR_1(y,z,t,t_1) + R^*_1(t_1)\right]  - \mQ(y,T_1,T_2)\right]\pd_z \Big[- \log D(T_1,T_2) + \mu^*(y,T_1,T_2) +	z-z_{T_1}\Big] \\
		 &= \lim_{z\to z_{T_1}} \int\limits_{T_1}^{T_2} dt_1 \left[ 	\mR_1(y,z,t,t_1) + R^*_1(t_1)\right]  - \mu^*(y,T_1,T_2)\,,
	\end{aligned}
\end{equation*}
where the shift operators in $\mR_1$ turn-out not to have any effect in this case, since the only term left after applying the derivative with respect to $z$ is independent of the state variables.

The subleading order at time $t$ is:
\begin{equation}
	\begin{aligned}
		V^{(0,1)}(y,t)  & = \int d\eta d\zeta V^{(1)}(\eta, T_1)  G_0(y,z,t;\eta,\zeta,T_1 )  \\
		 &= \int\limits_{T_1}^{T_2} dt_1 \left[ 	\mR_1(\phi(t,t_1)y,z,t,u) + R^*_1(u)\right]  - \mu^*(\phi(t,T_1)y,T_1,T_2)\,.
	\end{aligned}
\end{equation}
Again, the apparent $z$-dependence on the right-hand side is absent, since the shift operators in $\mR_1$ do not introduce any $z$-shift.
As with the 3M contract of App.~\ref{app:3M_PV}, the other first-order contribution to the fair value vanishes: $V^{(1,0)}(y,t)=0$.
This leads to the 1M-SOFR price of Eq. \eqref{eq:SOFR_1M}.

\section{Forward-looking futures price }\label{app:Eurodollar}
Here, we derive the fair value of the standard futures contract referencing forward-looking rates in the model of Sec. \ref{sec:ModelDescription}.\footnote{These contracts reference a LIBOR rate are equivalent to contracts referencing forward-looking term rates.}
In this case, the payoff depends only on state variables at $T_1$:
\begin{equation}
	\begin{aligned}
		P_{L}(y_{T_1},T_1,T_2) &= \frac{1}{F^{T_2}(y_{T_1},T_1) } -1 \sim  P^{(0)}(y_{T_1},T_1,T_2)  +  P^{(1)}(y_{T_1},T_1,T_2) \,,\\
		 P^{(0)}(y,T_1,T_2)  & \coloneqq \frac{e^{\mu^*(y,T_1,T_2)}}{ D(T_1,T_2)} - 1\,,\\
		 P^{(1)}(y,T_1,T_2) & \coloneqq \frac{e^{\mu^*(y,T_1,T_2)}}{ D(T_1,T_2)} \mF_1(T_1,T_2)\,.
	\end{aligned}
\end{equation}
Therefore, the fair value follows from a single convolution with the Green's function, Eq. \eqref{eq:G_compounded}, between $t$ and $T_1$:
\begin{equation} \begin{aligned}
PV(y,t) &= E\left[P_L(y_{T_1},T_1, {T_2})\Big| \mF_t\right] \sim V^{(0)}(y,T_1,T_2)  +  V^{(1)}(y,T_1,T_2) \,,\\
 V^{(0)}(y,T_1,T_2) & = \int d\eta d\zeta P^{(0)}(\eta,T_1,T_2) G_0(y,z,t; \eta,\zeta,T_1) \,,\\
  V^{(1)}(y,T_1,T_2) &= \int d\eta d\zeta P^{(1)}(\eta,T_1,T_2) G_0(y,z,t; \eta,\zeta,T_1)\,.
\end{aligned} \end{equation}
The non-trivial contributions to the fair value are:\footnote{The subleading contribution from the first-order Green's function, Eq. \eqref{eq:kernel_O1}, and $P^{(0)}(y,T_1,T_2)$ vanishes\\
	{$\int d\eta d\zeta P^{(0)}(\eta,T_1,T_2) G_1(y,z,t; \eta,\zeta,T_1)=0$}.}
\begin{equation}\label{eq:Eurodollar_V0}
	\begin{aligned}
		 V^{(0)}(y,T_1,T_2) & = \frac{e^{\mu^*(\phi(t,T_1)y,T_1,T_2)+ \half \tilde V_C(t,T_1,T_2)}}{D(T_1,T_2)} -1\,, \\
		 \tilde V_C(t,T_1,T_2) & \coloneqq \BTTs \Sigmarr(t,T_1)\,,
	\end{aligned}
\end{equation}
and the first order is
\begin{equation}\label{eq:Eurodollar_V1}
	\begin{aligned}
		V^{(1)}(y,t) &= \left\{ \int\limits_{T_1}^{T_2} dt_1 \left[\widetilde{R}_1^+(y,T_1,t_1,T_2) \widetilde{\mM}^+(y,t,T_1,t_1) -\widetilde{R}_1^-(y,T_1,t_1,T_2)\widetilde{\mM}^-(y,t,T_1,t_1) + R^*_1(t_1) \right] \right. \\
		&\left. \phantom{\int\limits_{T_1}^{T_2}} -\mu^*(\phi_r(t,T_1)y,T_1,T_2)- \tilde V_C(t,T_1,T_2) +\half \Sigmazz(T_1,T_2)\right\}  { e^{\mu^*(\phi_r(t,T_1)y,T_1,T_2)+ \half \tilde V_C(t,T_1,T_2) }\over D(T_1,T_2)}\,,
	\end{aligned}
\end{equation}
where the functions $\widetilde{R}^{\pm}_1$, $R^*_1(t_1)$  and the shift operators $\widetilde{\mM}^\pm(y,t,T_1,t_1)$ are those defined in \cite{AnalyticRFROptions} without the $\ \widetilde{\ }\ $ symbol:\footnote{We introduce $\ \widetilde{\ } \ $ here to avoid confusion with the rest of the definitions used in this paper.}
\begin{equation}\label{eq:operators_RFR_cap}
	\begin{aligned}
		\widetilde{R}_1^\pm(y,t,t_1,v) & \coloneqq {\psi_r(t,t_1)\over 2\gamma(t_1) }e^{\pm \gamma(t_1) \left(\phi_r(t,t_1)y + y^*(t_1) - B^+(t,t_1,v) \Sigmarr(t,t_1)-\Sigmarz(t,t_1) \right)} \\
		\widetilde{\mM}^\pm(y,t,T_1,t_1)  f(y) &= f(y+\widetilde{\Delta y}(t,T_1,t_1))\,,\\
		\widetilde{\Delta y}(t,T_1,t_1) &\coloneqq \gamma(t_1) {\phi_r(T_1,t_1) \over \phi_r(t,T_1)} \Sigmarr(t,T_1)\,.
	\end{aligned}
\end{equation}
The Eurodollar fair value incorporating market smile and skew is thus:
\begin{equation}\label{eq:V_Eurodollar}
	V_{Fwd-look.}(y=0,t=0)\sim  V^{(0)}(0,0) + V^{(1)}(0,0)\,,
\end{equation}
where $V^{(0)}$	 and $V^{(1)}$ are given in Eqs. \eqref{eq:Eurodollar_V0} and \eqref{eq:Eurodollar_V1}.
The convexity is obtained from the difference of Eq. \eqref{eq:V_Eurodollar} and Eq. \eqref{eq:Vfwd}.

\section{Forward- and Backward-looking futures price in the Hull-White model}\label{app:HW}
In this section, we present the formulae of the fair value of SOFR and forward-looking (Eurodollar or term-rates)
futures in the Hull-White model.
For comparison purposes, we calibrate the Hull-White model at the at-the-money implied volatility of caps \cite{AnalyticRFROptions}.

\subsection{SOFR futures}
The fair value of the SOFR futures in the Hull-White model corresponds to the zeroth-order terms of Theorems \ref{3MPrice} and \ref{1MPrice}.
This corresponds to $V^{(0)}(y,t)$ in Eq. \eqref{eq:SOFR_V0} for the continuously compounded contract (3M), and in Eq. \eqref{eq:SOFR_V0_1M} for the simple averaging contract (1M):
\begin{equation}\label{eq:V_HW_comp}
	V^{(HW)}_{Bwd-look.}(y,t)=
		\begin{cases}
			   \begin{aligned}
				   & - \log D(T_1,T_2) + \mu^*(y\phi(t,T_1),T_1,T_2)\,,    & T_2-T_1 = 1M\,,\\
					 &{e^{\mu^*(\phi(t,T_1)y,T_1,T_2)+ \half V_C(t,T_1,T_2)}\over D(T_1,T_2)} -1\,,     & T_2-T_1 = 3M\,.
			   \end{aligned}
		\end{cases}
\end{equation}
\subsection{Forward-looking futures}
The fair value of the forward-looking futures contract was calculated in Eq. \eqref{eq:Eurodollar_V0}:
\begin{equation}\label{eq:V_HW}
	V^{(HW)}_{Fwd-look.}(y,t)=	\frac{e^{\mu^*(\phi(t,T_1)y,T_1,T_2)+ \half \tilde V_C(t,T_1,T_2)}}{D(T_1,T_2)} -1\,.\\
\end{equation}

\end{document}